\documentclass{article}
\usepackage{graphicx}        % for figures
\usepackage{natbib}          % flexible author–year citations
\usepackage{amsmath,amsthm}
\usepackage{amsfonts} % or \usepackage{amssymb}
\usepackage{mathtools}
\usepackage{microtype} % helps with line breaks in text
\DeclareMathOperator{\Var}{Var}
\DeclareMathOperator{\Cov}{Cov}
\DeclareMathOperator{\plim}{plim}
\newtheorem{theorem}{Theorem}
\newcommand{\E}{\mathbb{E}}

\title{Note on Selection Bias in Observational Estimates of Algorithmic Progress}
\author{Parker Whitfill}
\date{August 2025}

\begin{document}

\maketitle

\section{Introduction}
\cite{ho2024algorithmic} is a very interesting paper that attempts to estimate the degree of algorithmic progress from language models. They collect observational data on language models' loss and compute over time, and argue that as time has passed, language models' algorithmic efficiency has been rising. That is, the loss achieved for fixed compute has been dropping over time. In this note, I want to raise one potential methodological problem with the estimation strategy. Intuitively, if part of algorithmic quality is latent (not observed in \cite{ho2024algorithmic}'s data), and compute choices are endogenous to algorithmic quality, then the estimation strategy in \cite{ho2024algorithmic} will not recover unbiased estimates of the true degree of algorithmic progress because of selection bias. 

\subsection{Background}
In \cite{ho2024algorithmic}, the loss $L$ is estimated as  
\begin{equation}\label{eq:scaling}
L = E + \frac{A}{(N q_N)^{\alpha}} + \frac{B}{(D q_D)^{\beta}}
\end{equation}
where:
\begin{itemize}
    \item $N$ is the number of model parameters,
    \item $D$ is the number of training data points,
    \item $q_D$ and $q_N$ are \emph{productivity factors} that capture how efficiently parameters or data are used
    \item $E$ is irreducible loss. 
\end{itemize}

To actually estimate this equation, \cite{ho2024algorithmic} assume that the productivity terms grow exponentially with calendar time. That is, equation 2 in \cite{ho2024algorithmic} sets 
\begin{equation}\label{eq:prod-time}
q_N(Y) = \exp\!\left(\alpha'(Y - Y_0)\right),\quad
q_D(Y) = \exp\!\left(\beta'(Y - Y_0)\right)
\end{equation}

where $Y$ is calendar year and $Y_0$ is the reference year. Substituting back into the scaling law yields the main estimation equation in \cite{ho2024algorithmic}, 
\begin{align}\label{eq:estimation}
L  &= E + \frac{A}{N^{\alpha}} e^{-\alpha\alpha'(Y - Y_0)}
   + \frac{B}{D^{\beta}} e^{-\beta\beta'(Y - Y_0)} \notag \\
   &\coloneqq E + \frac{A}{N^{\alpha}} e^{-\alpha_{\mathrm{year}}(Y - Y_0)}
   + \frac{B}{D^{\beta}} e^{-\beta_{\mathrm{year}}(Y - Y_0)} .
\end{align}

\subsection{Data-Generating Process}
Equation \ref{eq:prod-time} is hard to believe if taken literally. It says that algorithmic progress is a deterministic function of time with no heterogeneity across labs in the same year. This would, for example, rule out Anthropic having better algorithms than xAI, which in turn has better algorithms than academia in 2025. 

To accommodate this intuition, we can instead take equation \ref{eq:prod-time} as the best time-exponential approximation instead of holding exactly. That is, 
\begin{equation}
q_N(Y) = \exp\!\left(\alpha'(Y - Y_0) + \epsilon_N \right), 
\quad
q_D(Y) = \exp\!\left(\beta'(Y - Y_0) + \epsilon_D \right)
\end{equation}
where $\E[\epsilon_N] = \E[\epsilon_D] = \Cov(Y, \epsilon_D) = \Cov(Y, \epsilon_N) = 0$.\footnote{Formally, $\epsilon_N$ is defined as the residual from regressing $\ln q_N$ on $Y-Y_0$ and likewise for $\epsilon_D$. } $\epsilon_N, \epsilon_D$ capture within-year, across-lab algorithmic heterogeneity.

We will suppose the rest of the Data-Generating Process (DGP) follows equation \ref{eq:estimation} exactly. Therefore, the DGP is given by the following equation, 
\begin{equation}\label{DGP}
L 
= E + \frac{A}{(e^{\epsilon_N} N)^{\alpha}} \, e^{-\alpha_{\mathrm{year}} (Y - Y_0)}
+ \frac{B}{(e^{\epsilon_D} D)^{\beta}} \, e^{-\beta_{\mathrm{year}} (Y - Y_0)} .
\end{equation}

Suppose we estimate equation \ref{eq:estimation} from data generated by equation \ref{DGP}. Then our estimates of $\alpha, \beta$ can be biased by correlations between $\epsilon_N, \epsilon_D$ and $N, D$ for standard selection bias reasons. For example, if $D$ and $\epsilon_D$ are positively correlated, then when $D$ increases $L$ decreases from both $D$ increasing and $\epsilon_D$ increasing, but credit is all assigned to $D$, resulting in a biased estimate of $\beta$. 

Should we expect selection bias between the latent parts of algorithmic quality and compute usage? Standard economic logic would indicate yes: compute usage is chosen by labs endogenously after observing their \emph{algorithmic} quality (say on smaller experiments). 

\subsection{Formal Bias Analysis}
For simplicity, let's focus on the bias in just \emph{dataset size} by assuming $N \approx \infty$ and $E \approx 0$.\footnote{The same argument holds symmetrically for $N,\epsilon_N$. } This results in the following DGP
\[
L 
= \frac{B}{(e^{\epsilon_D}D)^{\beta}} \, e^{-\beta_{\mathrm{year}} (Y - Y_0)}.
\]

Taking logs results in 
\begin{equation}\label{DGP_simple}
    \ln L
= \ln B - \beta\ln D - \beta_{\mathrm{year}}(Y-Y_0) - \beta\epsilon_D.
\end{equation}

Similarly, taking logs of the estimation equation results in 
\begin{equation}\label{estimation_simple}
    \ln L
= \widehat{\ln B} - \hat{\beta} \ln D - \hat{\beta}_{\mathrm{year}}(Y-Y_0).
\end{equation}

Estimating equation \ref{estimation_simple} from DGP \ref{DGP_simple} now looks like the familiar omitted variable problem in econometrics. It is not guaranteed that $\Cov(\ln D, \epsilon_D) = 0$ which means our estimates may be biased by selection. 

Recall the ultimate object of interest is $\frac{\beta_{\mathrm{year}}}{\beta}$ which is the rate of algorithmic progress. The following theorem signs the bias on this object. 

\begin{theorem}
\label{thm:bias-sign}
Assume \(\Cov(\ln D, Y)>0\), \(\beta>0\), \(\beta_{\mathrm{year}}>0\), and
\[
\Cov(\ln D,\epsilon_D)\;>\;
-\Var\!\left(
  \ln D-\frac{\Cov(\ln D,Y)}{\Var(Y)}\,Y
\right).
\]
Define the large-sample bias from \eqref{estimation_simple} as
\[
\operatorname{bias}\!\left(\tfrac{\hat{\beta}_{\mathrm{year}}}{\hat{\beta}}\right)
:=\plim\!\left(\tfrac{\hat{\beta}_{\mathrm{year}}}{\hat{\beta}}\right)
  -\tfrac{\beta_{\mathrm{year}}}{\beta}.
\]
Then
\[
\operatorname{sign}\!\left(\operatorname{bias}\!\left(\tfrac{\hat{\beta}_{\mathrm{year}}}{\hat{\beta}}\right)\right)
= -\,\operatorname{sign}\!\bigl(\Cov(\ln D,\epsilon_D)\bigr).
\]
\end{theorem}

Note the first three assumptions of the theorem are very mild as they simply say dataset size has been growing over time, returns to dataset size are positive and algorithmic progress is positive. The final assumption is a technical assumption that lets us more cleanly sign the bias. Intuitively, it states that the selection between dataset size and algorithmic progress cannot be too negative. 

It's hard to know the sign of $\Cov(\ln D, \epsilon_D)$. One argument for it being positive is that firms with the largest datasets also tend to have the best algorithms (strong data-engineering pipelines co-move with algorithmic know-how). On the flip side, one argument for it being negative is that having better algorithms lessens the need to rely on data scaling. For example, suppose Anthropic and xAI are competing to release the best product. Suppose Anthropic has sufficiently better algorithms such that at modest data they can release a better model than xAI can release with massive data. If Anthropic cares only about releasing a leading model, then they have no incentive to scale data past modest amounts, since they are already winning with less data. Therefore, having better algorithms (high $\epsilon$) might be negatively correlated with dataset size (low $D$). 

Beyond the sign of the bias, the magnitude of the bias will depend on the severity of the selection. One suggestive piece of empirical evidence on this front is that $\beta$ is estimated experimentally in \cite{hoffmann2022training} and \cite{besiroglu2024chinchilla} at around $0.37$.\footnote{Since these experiments keep algorithms fixed, there is no possibility of selection bias.} However, \cite{ho2024algorithmic} estimate $\beta$ around $0.04$, potentially indicating negative bias in $\beta$.\footnote{See Table 2 in \cite{ho2024algorithmic}. Note that the original paper argues this discrepancy is due to the scale of language models in their dataset. } Since the estimate of algorithmic progress is $\frac{\beta_{\mathrm{year}}}{\beta}$, then a first-pass estimate is that algorithmic progress is overstated by around a factor of nine.\footnote{This is very much a first-pass argument as \cite{hoffmann2022training}'s and \cite{ho2024algorithmic}'s results could differ for a great many reasons beyond selection bias. }

\section{Monte Carlo Simulation}

To verify the theorem and quantify the bias magnitude, we conduct Monte Carlo simulations using the actual dataset sizes $D_i$ and years $Y_i$ from \cite{ho2024algorithmic}'s data. We generate data according to the DGP in equation \eqref{DGP_simple}, where we set $\beta = 0.37$ (from \cite{besiroglu2024chinchilla}) and $\beta' = 0.45$ (roughly half of what \cite{ho2024algorithmic} find).

We generate $\epsilon_D \sim N(0, \sigma^2)$ where $\sigma$ is set to half the mean of $\beta'(Y_i - Y_0)$. We vary the correlation between $\epsilon_D$ and $\ln D$ from -1 to 1 while ensuring $\text{Cov}(\epsilon_D, Y) = 0$. We then estimate the model in equation \eqref{estimation_simple} using OLS. We then compare the estimated parameters to the true parameters. 

Figure \ref{fig:mc_results} shows results from 1000 simulations for each correlation value. The simulations confirm the theorem: positive correlation leads to underestimating algorithmic progress, while negative correlation leads to overestimation. The bias is economically significant---with correlation 0.5, the true 45\% annual progress is estimated as only 16.5\%, while with correlation -0.5 it appears as 93\%. Of course, the actual numbers we get are sensitive to the distribution of $\epsilon_D$ that we pick, especially the variance. Accordingly, the exact mangitudes should not be taken too seriously. 

\begin{figure}[h!]
    \centering
    \includegraphics[width=\textwidth]{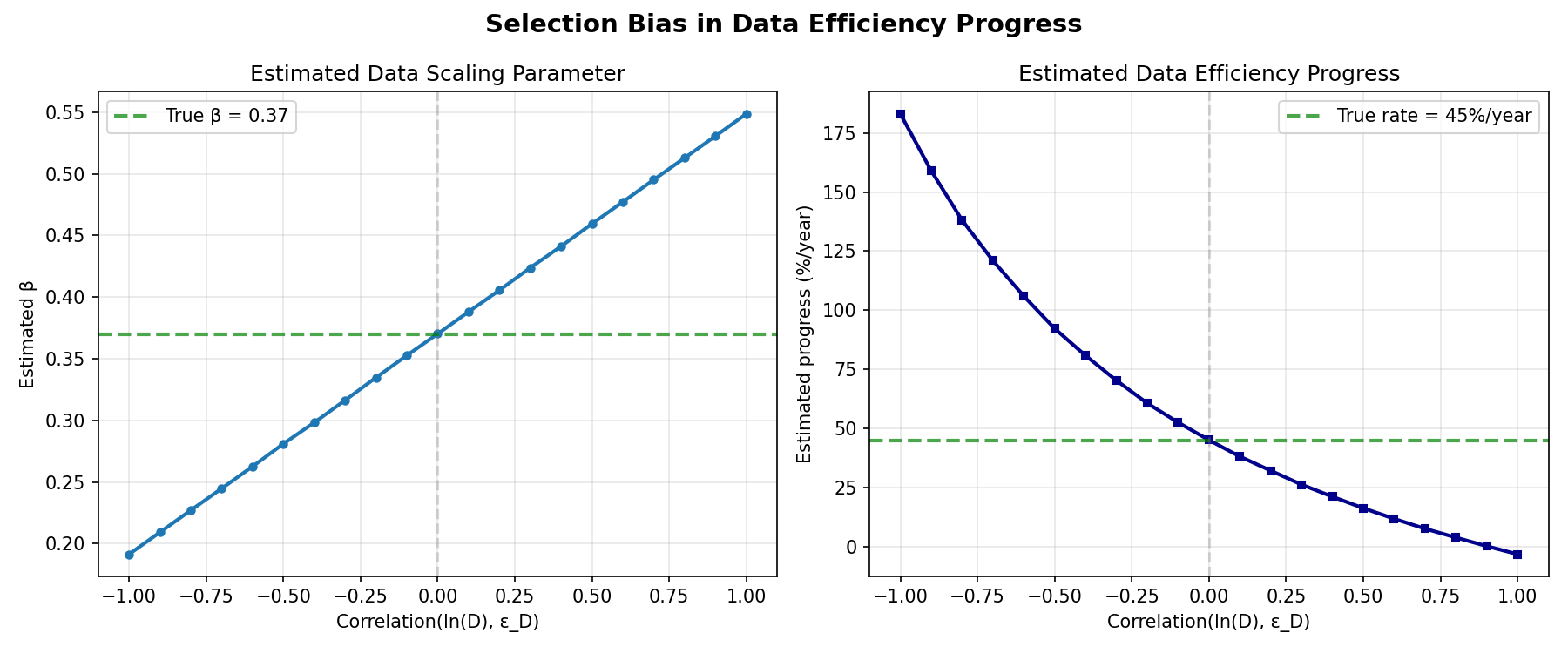}
    \caption{Monte Carlo Simulation}
    \label{fig:mc_results}
\end{figure}

\subsection{Conclusion}
Although this note has focused on \cite{ho2024algorithmic}, the general lesson of endogeneity of compute choices to algorithm quality applies to almost all attempts to infer algorithmic progress from observational data. 

One potential solution to this problem is to focus on experimental work where compute or algorithms can be randomly assigned. Alternatively, one could use observational data but find some plausible instrument that e.g., exogenously varies compute allocations but not algorithm quality. 

\section{Appendix}

% ===== Proof (align*, no extra numbering) =====
\begin{proof}[Proof of Theorem~\ref{thm:bias-sign}]

\textit{Step 1 (Population normal equations at the probability limits).}
\begin{align*}
\Cov(\ln D,\ln L)
&= -\,\plim\hat{\beta}\;\Var(\ln D)
   -\,\plim\hat{\beta}_{\mathrm{year}}\;
      \Cov(\ln D,Y-Y_0),\\
\Cov(Y-Y_0,\ln L)
&= -\,\plim\hat{\beta}\;
      \Cov(\ln D,Y-Y_0)
   -\,\plim\hat{\beta}_{\mathrm{year}}\;
      \Var(Y-Y_0).
\end{align*}

\textit{Step 2 (Evaluate covariances under the DGP \eqref{DGP_simple}).}
Using \(Y_0\) constant (so \(\Var(Y-Y_0)=\Var(Y)\) and \(\Cov(\ln D,Y-Y_0)=\Cov(\ln D,Y)\)) and \(\Cov(Y,\epsilon_D)=0\) by construction,
\begin{align*}
\Cov(\ln D,\ln L)
&= -\beta\,\Var(\ln D)
   -\beta_{\mathrm{year}}\,\Cov(\ln D,Y)
   -\beta\,\Cov(\ln D,\epsilon_D),\\
\Cov(Y-Y_0,\ln L)
&= -\beta\,\Cov(\ln D,Y)
   -\beta_{\mathrm{year}}\,\Var(Y).
\end{align*}

\textit{Step 3 (Equate and solve for probability limits).}
Equating Step 1 and Step 2 gives the linear system
\begin{align*}
\beta\,\Var(\ln D)
&+\beta_{\mathrm{year}}\,\Cov(\ln D,Y)
+\beta\,\Cov(\ln D,\epsilon_D)\\
&= \plim\hat{\beta}\,\Var(\ln D)
 + \plim\hat{\beta}_{\mathrm{year}}\,
    \Cov(\ln D,Y),\\[4pt]
\beta\,\Cov(\ln D,Y)
&+\beta_{\mathrm{year}}\,\Var(Y)\\
&= \plim\hat{\beta}\,\Cov(\ln D,Y)
 + \plim\hat{\beta}_{\mathrm{year}}\,\Var(Y).
\end{align*}
Eliminating \(\plim\hat{\beta}_{\mathrm{year}}\):
\begin{align*}
\plim\hat{\beta}\,
\Bigl(\Var(\ln D)\Var(Y)
      -\Cov(\ln D,Y)^2\Bigr)
&= \beta\,
   \Bigl(\Var(\ln D)\Var(Y)
         -\Cov(\ln D,Y)^2\Bigr)\\
&\quad+\beta\,\Var(Y)\,
         \Cov(\ln D,\epsilon_D).
\end{align*}
Hence
\begin{align*}
\plim\hat{\beta}
= \beta\!\left(
  1+\frac{\Var(Y)\,\Cov(\ln D,\epsilon_D)}
         {\Var(\ln D)\Var(Y)
          -\Cov(\ln D,Y)^2}
\right).
\end{align*}
Eliminating \(\plim\hat{\beta}\) instead gives
\begin{align*}
\plim\hat{\beta}_{\mathrm{year}}
= \beta_{\mathrm{year}}
  -\beta\,
   \frac{\Cov(\ln D,\epsilon_D)\,
         \Cov(\ln D,Y)}
        {\Var(\ln D)\Var(Y)
         -\Cov(\ln D,Y)^2}.
\end{align*}

\textit{Step 4 (Bias of the ratio; residual-variance rewrite).}
Combining the two limits and subtracting \(\beta_{\mathrm{year}}/\beta\),
\begin{align*}
\operatorname{bias}\!\left(\tfrac{\hat{\beta}_{\mathrm{year}}}{\hat{\beta}}\right)
&= -\,
   \frac{\Bigl(\beta_{\mathrm{year}}\Var(Y)
               +\beta\,\Cov(\ln D,Y)\Bigr)\,
         \Cov(\ln D,\epsilon_D)}
        {\beta\Bigl(\Var(\ln D)\Var(Y)
                    -\Cov(\ln D,Y)^2
                    +\Var(Y)\,
                     \Cov(\ln D,\epsilon_D)\Bigr)}\\[4pt]
&= -\left(
      \tfrac{\beta_{\mathrm{year}}}{\beta}
      +\tfrac{\Cov(\ln D,Y)}{\Var(Y)}
    \right)
    \cdot
    \frac{\Cov(\ln D,\epsilon_D)}
         {\Var\!\left(
           \ln D-\tfrac{\Cov(\ln D,Y)}{\Var(Y)}\,Y
          \right)
          +\Cov(\ln D,\epsilon_D)}.
\end{align*}

\textit{Step 5 (Sign).}
Because \(\beta>0\), \(\beta_{\mathrm{year}}>0\), and
\(\Cov(\ln D,Y)>0\),
\[
\tfrac{\beta_{\mathrm{year}}}{\beta}
+\tfrac{\Cov(\ln D,Y)}{\Var(Y)}
> 0.
\]
By the stated bound,
\[
\Var\!\left(
 \ln D-\tfrac{\Cov(\ln D,Y)}{\Var(Y)}\,Y\right)
+\Cov(\ln D,\epsilon_D)
> 0.
\]
Hence the denominator is positive and
\[
\operatorname{sign}\!\left(\operatorname{bias}\!\left(
\tfrac{\hat{\beta}_{\mathrm{year}}}{\hat{\beta}}\right)\right)
= -\,\operatorname{sign}\!\bigl(\Cov(\ln D,\epsilon_D)\bigr).
\qedhere
\]
\end{proof}

\bibliography{research}     % looks for research.bib in the same folder
\end{document}